\documentclass[10pt,a4paper]{IEEEtran}
\IEEEoverridecommandlockouts
\usepackage{cite}
\usepackage{amsmath,amssymb,amsfonts}
\usepackage{algorithmic}
\usepackage{graphicx}
\usepackage{epstopdf}
\usepackage{enumerate}
\usepackage{subfig}
\usepackage{textcomp}
\usepackage{amsthm}
\def\BibTeX{{\rm B\kern-.05em{\sc i\kern-.025em b}\kern-.08em
    T\kern-.1667em\lower.7ex\hbox{E}\kern-.125emX}}

\newcommand{\vect}[0]{\text{vect}}

\newcommand{\rank}[0]{\text{rank}}

\newcommand{\diag}[0]{\text{diag}}

\newcommand{\bgamma}[0]{\boldsymbol{\gamma}}

\newcommand{\bdelta}[0]{\boldsymbol{\delta}}

\newcommand{\btheta}[0]{\boldsymbol{\theta}}

\newcommand{\bSigma}[0]{\boldsymbol{\Sigma}}

\newcommand{\bA}[0]{\mathbf{A}}
\newcommand{\bb}[0]{\mathbf{b}}
\newcommand{\bB}[0]{\mathbf{B}}

\newcommand{\bC}[0]{\mathbf{C}}

\newcommand{\bD}[0]{\mathbf{D}}

\newcommand{\bE}[0]{\mathbf{E}}

\newcommand{\bg}[0]{\mathbf{g}}
\newcommand{\bG}[0]{\mathbf{G}}

\newcommand{\bH}[0]{\mathbf{H}}

\newcommand{\bI}[0]{\mathbf{I}}

\newcommand{\bJ}[0]{\mathbf{J}}

\newcommand{\bK}[0]{\mathbf{K}}

\newcommand{\bL}[0]{\mathbf{L}}

\newcommand{\bM}[0]{\mathbf{M}}
\newcommand{\bn}[0]{\mathbf{n}}

\newcommand{\bP}[0]{\mathbf{P}}

\newcommand{\br}[0]{\mathbf{r}}
\newcommand{\bR}[0]{\mathbf{R}}
\newcommand{\bs}[0]{\mathbf{s}}

\newcommand{\bU}[0]{\mathbf{U}}

\newcommand{\bx}[0]{\mathbf{x}}
\newcommand{\bX}[0]{\mathbf{X}}
\newcommand{\by}[0]{\mathbf{y}}

\newcommand{\bz}[0]{\mathbf{z}}
\newcommand{\bZ}[0]{\mathbf{Z}}

\newcommand{\bDt}[0]{\mathbf{\tilde{D}}}
\newcommand{\bgt}[0]{\mathbf{\tilde{g}}}

\newcommand{\bGt}[0]{\mathbf{\tilde{G}}}

\newcommand{\bxt}[0]{\mathbf{\tilde{x}}}

\newcommand{\byt}[0]{\mathbf{\tilde{y}}}
\newcommand{\bzt}[0]{\mathbf{\tilde{z}}}

\newcommand{\bSigmat}[0]{\boldsymbol{\tilde{\Sigma}}}

\newcommand{\bdeltat}[0]{\boldsymbol{\tilde{\delta}}}

\newcommand{\brh}[0]{\mathbf{\hat{r}}}

\newcommand{\bRh}[0]{\mathbf{\hat{R}}}

\newcommand{\bthetah}[0]{\boldsymbol{\hat{\theta}}}

\newcommand{\zeros}[0]{\mathbf{0}}
\newcommand{\ones}[0]{\mathbf{1}}

\newcommand{\MCE}[0]{\mathcal{E}}

\newcommand{\beq}{\begin{equation}}
\newcommand{\eeq}{\end{equation}}
\newcommand{\bea}{\begin{array}}
\newcommand{\ena}{\end{array}}

\newcommand{\DL}{\begin{dashlist}}
\newcommand{\DLE}{\end{dashlist}}

\newtheorem{lemma}{Lemma}

\begin{document}

\title{Efficient Calibration of Radio Interferometers Using Block LDU Decomposition  \\
\thanks{$1$ A. Mouri Sardarabadi (ammsa@astro.rug.nl) and L.V.E. Koopmans are affiliated with Kapteyn Astronomical Institute, University of Groningen, The Netherlands. $2$ Alle-Jan van der Veen is affiliated with Delft University of Technology, Delft, The Netherlands}
}

\author{Ahmad Mouri Sardarabadi$^1$, Alle-Jan van der Veen$^2$ and L\'eon V. E. Koopmans$^1$
}

\maketitle

\begin{abstract}
Having an accurate calibration method is crucial for any scientific research done by a radio telescope. The next generation radio telescopes such as the Square Kilometre Array (SKA) will have a large number of receivers which will produce exabytes of data per day. In this paper we propose new direction-dependent and independent calibration algorithms that, while requiring much less storage during calibration, converge very fast. The calibration problem can be formulated as a non-linear least square optimization problem. We show that combining a block-LDU decomposition with Gauss-Newton iterations produces systems of equations with convergent matrices. This allows significant reduction in complexity per iteration and very fast converging algorithms. We also discuss extensions to direction-dependent calibration. The proposed algorithms are evaluated using simulations.
\end{abstract}


\section{Introduction}

One of the key challenges for current and future radio--telescopes, such as LOFAR (LOw Frequency ARray) \cite{refIdLOFAR} and SKA (Square Kilometre Array) \cite{aj05skabook}, is the accurate calibration of the instrument with reasonable computational complexity. Modern radio telescopes consist of many receivers which can be large dishes or sub-arrays beamformed into a single element. The calibration problem for radio interferometers has already been addressed by several authors \cite{aj02calibration,wijnholds2009,kazemi01102013,mouri2014}. During a calibration cycle we use our current knowledge of the radio sources (fore example known from previous observations), to find the gains of the receivers. However, because a typical interferometer has a direction-dependent behavior, we need to solve these gains for different directions \cite{yatawatta2015distributed,vandertol2007}.  In this paper we assume to have access to an accurate model for the sources and we are interested in developing computationally efficient algorithms that scale well with multi-channel observations.

Based on the resolution of the instrument, in order to avoid source smearing, the observations are divided into small snapshots (order of seconds). However, in order to study very weak sources we need to observe for a very long time (e.g. hundreds of hours). This, combined with a large number of channels (several hundreds), produces a substantial volume of data that needs to be processed. Also, because calibration is a non-linear and non-convex problem, iterative and alternating approaches usually form the basis for a practical solution \cite{wijnholds2009,yatawatta2015distributed}.

In this paper we use the Khatri--Rao structure of the matrices involved in data model to develop an efficient direction independent gain calibration algorithm. We then use this method as a building block for a direction dependent calibration algorithm. Additionally, for very large problems we propose a conjugate gradient based algorithm and use simulation to evaluate the performance of these methods.

\section{Data Model}
In this section we introduce the covariance model for the data. We assume to have access to $P$ (single polarization or unpolarized) receivers which are exposed to $N_s$ (compact/point) sources. We assume that sources can be grouped into $Q$ clusters which are affected by the same direction-dependent gain similar to the model presented in \cite{kazemi01102013}. We stack the voltage output of each receiver in a vector denoted by $\by$ and assume that narrow-band assumptions hold. This allows us to model the sampled output of the array as
\begin{equation*}
\by[n] = \sum_{q=1}^Q \bG_q \bs_q[n] + \bn[n] 
\end{equation*}
where $\bs_q[n]$ represents the total signal from the $q$th cluster which includes the array response, $\bG_q = \diag(\bg_q)$ is the common gain for the $q$th cluster and $\bn[n]$ is the noise of each receiver. The covariance matrix for this model is given by $\MCE\{\by\by^H\}$. However, we assume that some of the elements of this matrix are contaminated and/or are removed. We use a masking matrix $\bM$ containing zeros and ones to capture this missing data in the model. We also assume that the gains are stable over several ``snapshots" in both time and frequency. We assume to have $K$ frequency channels with $T$ snapshots each. Including the masking matrix we get the following covariance model for each snapshot 
\begin{equation}
\label{eq:data_model}
\bR_{t,k} =\bM \odot \MCE\{\by_{t,k}\by_{t,k}^H\} = \bM \odot \sum_{q=1}^Q \bG_q \bSigma_{q,t,k} \bG_q^H,
\end{equation}
where $t =1, \dots, T$, $k=1, \dots, K$, $\odot$ is the element-wise or Hadamard product, ${}^H$ is the Hermitian transpose and $\bSigma_{q,t,k} = \MCE\{\bs_{q,t,k}\bs_{q,t,k}^H\}$ is the covariance of the $q$th cluster or the ``predicted sky-model", which is assumed to be known. We also assume that $\bR_{\bn} = \MCE\{\bn\bn^H\}$ is diagonal and is always removed as a result of applying the mask matrix, $\bM$.

During the measurements, a noisy estimate of $\bR_{t,k}$ is made using the output of the receivers. This estimate is denoted as a sample covariance matrix or sampled visibilities and is given by
\begin{equation}
\bRh_{t,k} = \bM \odot \sum_{n=1}^N \by_{t,k}[n]\by_{t,k}[n]^H,
\end{equation}
where $N$ is the number of (voltage) samples in a single snapshot.

For the rest of this paper we stack the covariance (visibility) model and the data, respectively, into vectors
\begin{equation}
\br(\btheta) = \begin{bmatrix}
\vect(\bR_{1,1}) \\
\vect(\bR_{2,1}) \\
\vdots\\
\vect(\bR_{T,K})
\end{bmatrix}, \brh = \begin{bmatrix}
\vect(\bRh_{1,1}) \\
\vect(\bRh_{2,1}) \\
\vdots\\
\vect(\bRh_{T,K})
\end{bmatrix}
\end{equation}
where $\vect(.)$ produces a vector from the argument matrix by stacking its columns and
\[
\btheta = \begin{bmatrix}
\bg_1^T& \bg^H_1 & \dots & \bg^T_Q & \bg^H_Q
\end{bmatrix}^T
\]
is the ``augmented" vector of variables. The term augmented means that a complex variable and its conjugate are used as separate variables instead of the real and imaginary part of the complex variables.

Using this data model we want to estimate the gains for each direction.
\section{Direction Independent Algorithm}
\label{sec:DIC}
In this section we discuss the case where $Q=1$. In this case the entire available sky-model is used and as a result the gain solutions are assumed direction independent. This case forms the basis for the direction-dependent calibration, which is discussed in the next section.

We use the least squares cost function to find an estimate for the gains:
\begin{equation}
\bthetah = \arg\min_{\btheta} \|\brh - \br(\btheta)\|_2^2
\end{equation}
where $\|.\|_2$ is the $l_2$ norm of a vector.
Because of the non-linear and non-convex nature of this problem we use a Newton-based iterative method known as the Gauss-Newton algorithm. The updates for this algorithm are given by
\begin{equation}
\label{eq:GNupdate}
\bthetah^{(i+1)} = \bthetah^{(i)} + \mu^{(i)} \bdelta
\end{equation}
where the GN direction of descent $\bdelta$ is given by the solution of \cite{Frandsen2004}
\begin{equation}
\label{eq:GNdod}
\bJ^H\bJ \bdelta = \bJ^H\left [\brh - \br(\btheta) \right ]
\end{equation}
where
\begin{equation}
	\bJ = \frac{\partial \br(\btheta)}{\partial \btheta^T} = \begin{bmatrix}
		\bJ_{1,1}^T&\dots&\bJ_{T,K}^T
	\end{bmatrix}^T,
\end{equation}
\begin{equation}
\bJ_{t,k} = \bP\begin{bmatrix}
\bG^*\bSigma_{t,k}^T \circ \bI_P & \bI_P \circ \bG\bSigma_{t,k}
\end{bmatrix},
\end{equation}
with $\circ$ the Khatri-Rao product, ${}^*$ the complex conjugate and $\bP = \diag(\vect(\bM))$ a projection matrix corresponding to the mask matrix $\bM$.
There exists a phase ambiguity for the solutions, i.e. if $\bg$ is a solution so is \mbox{$\bg' = e^{i\phi}\bg$} for any real $\phi$. We call the problem identifiable if $\rank(\bJ)= 2P - 1$ where the deficiency by 1 is the result of the phase ambiguity. In this case a basis for the null space of $\bJ$ is given by 
\begin{equation}
\label{eq:null_J}
\bz = [\bg^T , -\bg^H]^T.
\end{equation}
Because \mbox{$\br(\btheta) = 1/2 \bJ\btheta$} we have \mbox{$\bJ^H\bJ (\bdelta + 1/2 \btheta) = \bJ^H\brh$} which combined with the fact that $\btheta^H\bz = \zeros$ and hence $\btheta$ is in the row space of $\bJ$, leads to
\begin{equation}
\bthetah^{(i+1)} = \left (1 - \frac{\mu^{(i)}}{2} \right )\bthetah^{(i)} + \mu^{(i)} \bdeltat
\end{equation}
which is equivalent to \eqref{eq:GNupdate} for $\bdeltat$ satisfying
\begin{equation}
\label{eq:DoD_system}
\bJ^H\bJ \bdeltat = \bJ^H\brh.
\end{equation}
With this change of variables for the direction of descent, we remove the necessity to update the model, $\br(\btheta)$. However, since $\bJ$ depends on $\btheta$, this is only beneficial if we can calculate operations involving $\bJ$ and $\bJ^H$ sufficiently fast.  
Calculating the models $\bSigma_{t,k}$ which are needed for calculating $\bJ$ is very expensive and we would like to pre-calculate these matrices only once. However, because $TK$ is large, storing all of these model matrices should also be avoided. The rest of this section focuses on solving \eqref{eq:DoD_system}, while avoiding storage of the sky-models $\bSigma_{t,k}$.

For square matrices $\bA$ and $\bB$ we have \mbox{$\bI \circ (\bA\odot\bB) = \diag(\vect(\bB)) (\bI \circ \bA)$} and \mbox{$(\bB^T\odot\bA) \circ \bI = \diag(\vect(\bB)) (\bA \circ \bI)$} . Using these relations we have 
\[
\bJ_{t,k} = \begin{bmatrix}
\bG^*(\bM \odot \bSigma_{t,k})^T \circ \bI_P & \bI_P \circ \bG(\bM \odot \bSigma_{t,k})
\end{bmatrix}.
\]
Combining these results with \mbox{$\bJ^H\bJ = \sum_k\sum_t \bJ_{t,k}^H\bJ_{t,k}$} and \mbox{$\bJ^H\brh = \sum_k\sum_t \bJ_{t,k}^H\brh_{t,k}$}
we have
\begin{equation}
\bJ^H\bJ = \begin{bmatrix}
\diag\left [ \bH ~(\bg \odot \bg^*) \right ] & \bG \bH \bG \\
\bG^* \bH \bG^* & \diag\left [ \bH ~(\bg \odot \bg^*) \right ]
\end{bmatrix}
\end{equation}
and 
\begin{equation}
\bJ^H\brh = \begin{bmatrix}
\bE \bg\\
\bE^* \bg^*
\end{bmatrix}
\end{equation}
where
\begin{align}
\label{eq:di_H}
\bH & = \bM \odot \sum_k\sum_t \bSigma_{t,k}^T \odot \bSigma_{t,k},\\
\label{eq:di_E}
\bE &= \bM \odot \sum_k\sum_t \bSigma_{t,k}^T \odot \bRh_{t,k}.
\end{align}
We only need to calculate the real symmetric matrix $\bH$ and the Hermitian matrix $\bE$ once in order to solve $\bdeltat$ and $\bthetah$. This means that we can discard $\bSigma_{t,k}$ during the calculation of $\bH$ and $\bE$. This allows for a dramatic reduction of the required storage and also I/O overhead during the calibration.

The remaining problem is the actual solution of \eqref{eq:DoD_system} which we address now. We would like to point out that this system of equations is normal and consistent. This allows for the solution to be obtained from
\[
\bdeltat = \bX\bJ^H\brh,
\]
where $\bX$ is any generalized inverse of $\bJ^H\bJ$ (i.e. \mbox{$\bJ^H\bJ\bX\bJ^H\bJ = \bJ^H\bJ$}).
However, not all $\bdeltat$ found in this way will have the augmented form \mbox{$[\by^T, \by^H]^H$}, which is required for a valid direction of descent. We use the following lemma to find a simple solution for this problem.
\begin{lemma}
	Let $\bK$ be a permutation matrix of the form
	\[
	\bK = \begin{bmatrix}
	\zeros & \bI_M \\
	\bI_M & \zeros
	\end{bmatrix},
	\]
	and $\bA$ be any square matrix of size $2M \times 2M$ such that \mbox{$\bA^* = \bK \bA \bK$}. Let $\bA^g$ be a generalized inverse of $\bA$ (i.e. \mbox{ $\bA\bA^g\bA = \bA$}) then \mbox{$\bX = \frac{1}{2} (\bA^g  + \bK (\bA^g)^* \bK)$} is also a generalized inverse of $\bA$.
\end{lemma}
	\begin{proof}
	The proof is a simple verification:
	\[
	\begin{array}{rl}
	\bA\bX\bA &= \frac{1}{2} (\bA\bA^g\bA) + \frac{1}{2} (\bA\bK (\bA^g)^* \bK\bA)\\
	& = \frac{1}{2}\bA + \frac{1}{2} (\bK\bK\bA\bK (\bA^g)^* \bK\bA\bK\bK) \\ 
	& = \frac{1}{2}\bA + \frac{1}{2} (\bK (\bA \bA^g \bA)^*\bK) = \bA,
	\end{array}
	\]
	where we used $\bK\bK = \bI$.
\end{proof}
It is trivial to verify that \mbox{$\bK\bJ^H\bJ\bK = (\bJ^H\bJ)^*$} and  \mbox{$\bK\bJ^H\brh = (\bJ^H\brh)^*$}. This allows us to show that for any generalized inverse solution \mbox{$\bdeltat_1 = (\bJ^H\bJ)^g\bJ^H\brh$},
	\begin{equation}
	\label{eq:deltat1_to_deltat}
\bdeltat = \bX \bJ^H\brh= \frac{1}{2}( \bdeltat_1 +  \bK \bdeltat_1^*) 
\end{equation}
is a solution to the system of equation with the correct format. 

Based on this discussion, it is always possible to transform any solution to the correct (augmented) format. This gives us more flexibility in choosing our solver. For the matrix $\bJ^H\bJ$ we will show that using a block LDU decomposition will lead to solving a system of equations which involves a convergent matrix which has a stable and fast iterative solution \cite{meyer1977}.

In order to simplify the notation we introduce the following definitions: \mbox{$\bgt \equiv \bg^* \odot \bg$}, \mbox{$\bD \equiv \diag(\bH\bgt)$}, \mbox{$\bb \equiv \bD^{-1/2}\bE\bg$} and \mbox{$\bC \equiv \bD^{-1/2} \bG^* \bH \bG^* \bD^{-1/2}$}. With these definitions the block-LDU decomposition of  \mbox{$\bJ^H\bJ = \bL \bDt \bL^H$} is given by
\begin{equation*}
\begin{array}{cc}
\bL = \begin{bmatrix}
\bD^{1/2} & \zeros \\
\bD^{1/2}\bC & \bD^{1/2}
\end{bmatrix} & \text{and } \bDt = \begin{bmatrix}
\bI & \zeros \\
\zeros & \bI -  \bC\bC^H
\end{bmatrix}.
\end{array}
\end{equation*}
Applying forward-backward substitution we find the following expression for $\bdeltat_1$ in \eqref{eq:deltat1_to_deltat}: 
\begin{equation}
	\label{eq:deltat1}
\bdeltat_1 = \begin{bmatrix}
	\bD^{-1/2}(\bb - \bC^H \bdeltat_{1,2}) \\
	\bD^{-1/2}\bdeltat_{1,2}
\end{bmatrix}, 
\end{equation}
where $\bdelta_{1,2}$ is the solution to the following system of equations
\begin{equation}
	\label{eq:BLDU_equation}
(\bI - \bC\bC^H) \bdeltat_{1,2} = \bb^* - \bC\bb.
\end{equation}
Remembering that $\bJ^H\bJ$ is rank-deficient by one and the fact that $\bL$ is positive definite, we know that \mbox{$\bI - \bC\bC^H$} is also rank deficient by one and positive semidefinite. We already discussed that $\bz$ given by \eqref{eq:null_J} is a basis for the null space of $\bJ^H\bJ$. This means that \mbox{$\bL^H\bz$} is a basis for the null space of $\bDt$ and hence 
\[
\bzt = \frac{1}{\sqrt{\bg^H\bD\bg}}\bD^{1/2}\bg^*
\]
is a unit-norm basis for the null space of \mbox{$\bI - \bC\bC^H$}. Because the system of equations in \eqref{eq:DoD_system} is consistent, so is \eqref{eq:BLDU_equation} and
\[
\begin{array}{rl}
\bdeltat_{1,2} &= (\bI - \bC\bC^H)^\dagger (\bb^* - \bC\bb) \\
& =  (\bI - \bC\bC^H + \bz\bz^H)^{-1} (\bb^* - \bC\bb).
\end{array}
\]
Note that \mbox{$\bI - \bC\bC^H + \bzt\bzt^H$} is positive definite with \mbox{$\lambda_{\max} = 1$}, which means that the spectral radius of \mbox{$\rho(\bz\bz^H - \bC\bC^H) < 1$} and hence this matrix is a convergent matrix. For convergent matrices we know \cite{meyer1977} that
\begin{equation}
\label{eq:iterative_update}
\bdeltat_{1,2}^{(j+1)} = \bb^* - \bC\bb - (\bz\bz^H - \bC\bC^H) \bdeltat_{1,2}^{(j)}
\end{equation}
will converge to a solution of \eqref{eq:BLDU_equation}.

To summarize, in order to find a solution to \eqref{eq:DoD_system}, first we need to calculate $\bD$, $\bC$, $\bzt$ and then use \eqref{eq:iterative_update}, \eqref{eq:deltat1} and \eqref{eq:deltat1_to_deltat}. The complexity of these operations are $P$ divisions and $O(P^2)$ operations needed for the matrix vector multiplications. This means that we will benefit from the fast convergence of the GN algorithm, while having the same complexity as slower converging alternating algorithms.

The only unsolved issue is the optimal step-size $\mu^{(i)}$ which, as we show in Appendix \ref{app:stepzie}, requires solving for the roots of a third order polynomial with real coefficients, for which closed-form solutions exists.

\section{Extension to Direction Dependent Calibration}
\label{sec:DDC}
Now that several key ideas have been derived for the direction independent scenario, we extend to the direction-dependent case. Again we use the least squares cost function to find an estimate for the gains
\[
\bthetah = \arg \min_{\btheta}  \| \br - \sum_{q=1}^Q \br_q(\btheta_q) \|_2^2
\]
where 
\[
\br_q(\btheta_q) = \begin{bmatrix}
\vect(\bG_q \bSigma_{q,1,1}\bG_q)\\
\vdots \\
\vect(\bG_q \bSigma_{q,T,K}\bG_q)
\end{bmatrix}.
\]
Using this cost function we discuss two different approaches for solving this problem. The first one is based on the repeated application of the method developed for the direction independent scenario which we will denote as ``Block Gauss-Newton" (BGN) and the second approach which is based on the Conjugate Gradient (CG) method.

\subsection{Block Gauss-Newton}
We can extend the matrices $\bH$ and $\bE$ defined by \eqref{eq:di_H} and \eqref{eq:di_E} to the direction dependent case as 
\begin{align}
\bE_{q} & = \sum_t \sum_k \bSigma_{q,t,k}^T \odot \bRh_{t,k} \\
\bH_{{q_{1}},{q_{2}}} & = \sum_t \sum_k \bSigma_{{q_{1}},t,k}^T \odot \bSigma_{{q_{2}},t,k}
\end{align}
where $q$, $q_1$ and $q_2$ take values $1, \dots, Q$,  $\bH_{{q_{1}},{q_{2}}}$ is Hermitian and \mbox{$\bH_{{q_{1}},{q_{2}}} = \bH_{{q_{2}},{q_{1}}}^T$}. The use of these matrices is beneficial only if $KT/Q > 1$. If this condition does not hold, storing $\bRh_{t,k}$ and $ \bSigma_{q,t,k}$ will be more efficient than generating $\bE_q$ and $\bH_{q_1,q_2}$. We assume that this condition holds for a practical calibration scenario.

Using $\bE_q$ and $\bH_{{q_{1}},{q_{2}}}$, the gradient for the $q$th direction can be written as
\begin{equation}
\label{eq:grad_ddc}
\bgamma_q = \bJ_q^H(\brh - \br(\btheta)) = \begin{bmatrix}
\bE_q \bg_q - \sum_{q_2} \bG_{q_2} \bH_{q,q_2} (\bg_{q} \odot \bg^*_{q_2}) \\
\bE_q^T \bg^*_q - \sum_{q_2} \bG^*_{q_2} \bH_{q,q_2}^T (\bg^*_{q} \odot \bg_{q_2})
\end{bmatrix}.
\end{equation}
If we change the summation above such that $q_2 \neq q$, then the direction of descent can be found by applying the method discussed in the previous section separately for each direction in a parallelized fashion. Because the updates are done separately for each direction, we are not limited to a single iteration and we can update each solution several time before updating the gradients. This approach is very similar to the ADMM \cite{boyd2011distributed}. However, the calibration problem is not convex and the convergence of BGN is local.

\subsection{Conjugate Gradient}
For the next generation radio telescopes, such as the SKA, the number of stations and directions will increase dramatically. In these cases where the problem becomes very large even the modified Gauss-Newton method used in previous section could become prohibitive. Simple classical methods such as Conjugate Gradient (CG) become attractive in these scenarios. The CG has very nice convergence properties if the exact optimal step-size is used \cite{Frandsen2004}. If $KT/ Q > 1$ we can use the matrices $\bE_q$ and $\bH_{{q_{1}},{q_{2}}}$ to find the optimal step-size (see Appendix \ref{app:stepzie}). This, in combination with the Polak-Ribi\`ere method \cite{Frandsen2004} will produce a relatively fast converging CG method for the direction-dependent calibration.

For this algorithm we use the previous direction of descent and the gradient given by \eqref{eq:grad_ddc} to a new direction of descent. The updates for the direction of descent are
\[
\bdelta^{(j)} = \bgamma^{(j)} + \lambda \bdelta^{(j-1)}
\]
where \mbox{$\bgamma = [\bgamma_1^T, \dots, \bgamma_Q^T]$} and $\lambda$ is given by Polak-Ribi\`ere ratio
\[
\lambda = \frac{\Re\{(\bgamma^{(j)} - \bgamma^{(j-1)})^H\bgamma^{(j)}\}}{{\bgamma^{(j-1)}}^H\bgamma^{(j-1)}},
\]
where $\Re\{.\}$ is the real part of the argument. Using simulations we show that this algorithm is computationally competitive with other methods.

\section{Extension to Polarized Instrumental Gains}
For an array with $P$ receivers with dual polarization (e.g. a cross-dipole antenna) and direction independent gains, under narrow-band assumptions, we have a covariance model of the form
\begin{equation}
\bR = \bM \odot \bGt\bSigmat\bGt^H
\end{equation}
where
\begin{eqnarray*}
\bGt = \begin{bmatrix}
\bG_{x} & \bG_{xy} \\
\bG_{yx} & \bG_{y}
\end{bmatrix}, \text{ and }  \bSigmat = \begin{bmatrix}
\bSigma_{xx} & \bSigma_{xy} \\
\bSigma_{xy}^H & \bSigma_{yy}
\end{bmatrix},
\end{eqnarray*}
are the gain and source covariance matrix respectively.
Let the Jones matrix for a x-y pair be
\[
\begin{bmatrix}
g_{x,p} & g_{xy,p} \\
g_{yx,p} & g_{y,p}
\end{bmatrix},
\]
for $p = 1,\dots,P$. We stack each element into a vector such that 
\[
\bg_x = \begin{bmatrix}
g_{x,1}^T&
g_{x,2}^T&
\dots &
g_{x,P}^T
\end{bmatrix}^T.
\]
We do the same to create $\bg_{xy}$, $\bg_{yx}$ and $\bg_{y}$. With these definitions $\bG_{x} = \diag(\bg_{x})$, $\bG_{xy} = \diag(\bg_{xy})$ and so on.
Similar to \cite{kazemi01102013}, we can extend this model to direction-dependent case using
\begin{equation}
\bR = \bM \odot \sum_{q=1}^{Q} \bR_q= \bM \odot\sum_{q=1}^{Q}
\bGt_q\bSigmat_q\bGt_q^H.
\end{equation}
There are some simplifications that can be made to this model if the sources used in the calibration are unpolarized. This is discussed in the next section.

\subsection{Unpolarized Sources and Unitary Ambigouity}
For unpolarized sources $\bSigma_{xy} = \zeros$ and $\bSigma_{x} = \bSigma_{y}=\bSigma$. Hence,
\[
\bSigmat_q = \bI_2 \otimes \bSigma_q
\]
and without loss of generality we can find another gain  
\[
\bGt_q' = \bGt_q (\bU \otimes \bI_P)
\]
for any $2\times2$ unitary matrix $\bU$, such that
\[
\bGt_q'\bSigmat_q(\bGt_q')^H= \bGt_q\bSigmat_q\bGt_q^H.
\]
This is the so called unitary ambiguity problem which happens when all the sources in the cluster $q$ are unpolarized. There are several ways to address this ambiguity, fore example \cite{yatawatta2013radio} constraints the solutions on a manifold. In this section, we give an expression for the null space of the Jacobian as the result of this ambiguity.  

Because $\bU$ is unitary of size $2 \times 2$, it can be fully described by 4 independent real parameters (4 complex parameters with 4 constraints from $\bU^H\bU = \bI$). This means that 4 additional constraints per direction is needed to make the gains unique. This also means that the Jacobian,
\[
\bJ = \frac{\partial \vect(\bR)}{\partial \btheta^T},
\]
where $\btheta = [\btheta_1^T, \dots, \btheta_Q^T]^T$ with
\[
\btheta_q = [
\bg_{x,q}^T, \bg_{yx,q}^T, \bg_{xy,q}^T, \bg_{y,q}^T, \bg^H_{x,q}, \bg^H_{yx,q}, \bg^H_{xy,q}, \bg^H_{y,q}
]^T,
\] 
is rank deficient by at least $4Q$. As a result, there exits a $8PQ \times 4Q$ matrix $\bZ$ such that $\bJ\bZ = \zeros$. By some algebra we can show that if the receivers are not linearly polarized and/or there are no defective receivers, a basis for the null space of $\bJ$ can be constructed using the following relations: 
\[
\bZ =\begin{bmatrix}
\bZ_1 &  \zeros& \dots & \zeros\\
\zeros & \bZ_2 & \ddots & \vdots\\
\vdots &\ddots & \ddots& \zeros\\
\zeros& \dots &\zeros & \bZ_Q
\end{bmatrix}
\] 
with
\begin{equation}
\bZ_q = \begin{bmatrix}
\bg_{x,q} &        \zeros &        \zeros &   \bg_{xy,q}\\
\bg_{yx,q} &        \zeros &        \zeros &    \bg_{y,q}\\
\zeros &    \bg_{xy,q} &      \bg_{x,q}&      \zeros \\
\zeros &     \bg_{y,q} &    \bg_{yx,q} &      \zeros \\
-\bg^*_{x,q} &        \zeros & -\bg^*_{xy,q} &      \zeros \\
-\bg^*_{yx,q} &        \zeros &  -\bg^*_{y,q} &      \zeros \\
\zeros & -\bg^*_{xy,q} &        \zeros &  -\bg^*_{x,q} \\
\zeros &  -\bg^*_{y,q} &        \zeros & -\bg^*_{xy,q} \\
\end{bmatrix}.
\end{equation}
It is important to note that the columns of this basis are not necessary orthogonal.  For the rest of this paper we assume to have a sky-model consisting of unpolarized compact and/or point sources.

\section{Direction Independent Polarized Calibration with Unpolarized Sky-Model}
\label{sec:DIPOLCALUNM}
In this section we show how a polarized calibration problem with unpolarized sky-model can be transformed into an unpolarized direction-dependent gain calibration similar to Sec.~\ref{sec:DDC} with twice the number of directions.

Using the data model in previous section we have
\begin{align}
\bR &= \bM \odot \bGt (\bI_2\otimes \bSigma) \bGt^H \notag\\
	&= \bM \odot \notag \\
	&\begin{bmatrix}
	\bG_{x}\bSigma\bG_{x}^H + \bG_{xy}\bSigma\bG_{xy}^H & \bG_{x}\bSigma\bG_{yx}^H + \bG_{xy}\bSigma\bG_{y}^H \\
	\bG_{yx}\bSigma\bG_{x}^H + \bG_{y}\bSigma\bG_{xy}^H & \bG_{yx}\bSigma\bG_{yx}^H + \bG_{y}\bSigma\bG_{y}^H
	\end{bmatrix} \notag\\
	& = \bM \odot \left[\bG_1 (\ones_2\ones_2^T \otimes \bSigma)\bG_1^H + \bG_2 (\ones_2\ones_2^T \otimes \bSigma)\bG_2^H\right]\notag
\end{align}
where
\begin{eqnarray*}
\bG_1 = \begin{bmatrix}
\bG_{x} & \zeros \\
\zeros & \bG_{yx}
\end{bmatrix} & \text{and } \bG_2 = \begin{bmatrix}
\bG_{xy} & \zeros \\
\zeros & \bG_{y}
\end{bmatrix}.
\end{eqnarray*}
We see that the polarized model is now transformed into an unpolarized model with two ``directions" having the same sky model $(\ones_2\ones_2^T \otimes \bSigma)$. Note, that this change of representation does not remove the unitary ambiguity and while a direction-dependent unpolarized calibration with two directions has only a rank deficiency of 2, this problem still has a rank deficiency of 4. In fact the basis for the null space of the Jacobian is given by
\begin{equation}
\bZ = \begin{bmatrix}
\bg_{1} &        \zeros &         \bg_{2} &   \bg_{2}\\
-\bg^*_{1} &        \zeros & -\bg^*_{2} &      \bg^*_{2} \\
\zeros &    \bg_{2} &      \bg_{1}&     \bg_{1} \\
\zeros & -\bg^*_{2} &        -\bg^*_{1} &  \bg^*_{1} \\
\end{bmatrix},
\end{equation}
where $\bG_1 = \diag(\bg_{1})$ and $\bG_2 = \diag(\bg_{2})$.

As we argued before, because we are missing two additional dimensions in the Jacobian, the transformation of the polarized problem to an unpolarized one leads to an unidentified direction-dependent problem if each direction is solved independently (even when the polarized problem is identifiable). Additionally $\bg_{xy}$ and $\bg_{yx}$ may have zero elements, which in turn could make $\bg_1$ and $\bg_2$ unidentifiable for the methods in Sec.~\ref{sec:DDC}. For these reasons, and because the algorithm used for DI polarized calibration is the basis for direction-dependent case, we will propose a modified algorithm specially tailored for this model in the next section.

\subsection{Levenberg-Marquardt for DI Polarized Calibration}
The Levenberg-Marquardt alrogirm (LMA), is a popular extension to the Gauss-Newton algorithm where the length of the direction of descent, $\bdelta$, is regularized and the system of equations in \eqref{eq:DoD_system} is replaced by
\begin{equation}
\label{eq:LMupdates}
(\bJ^H\bJ + \lambda \bI) \bdelta = \bJ^H[\brh - \br(\btheta)]
\end{equation}
where $\lambda > 0$ is updated during each iteration (see \cite{Frandsen2004} for details on this update). The full treatment of the LMA is beyond the scope of this paper and we only focus on adaptation of our unpolarized direction-dependent approach to this problem. Because the adding the term $\lambda \bI$ already guarantees the positive definiteness of the matrices involved for the LMA, we do not need to use the $\bZ$ matrix.

Following the same procedure as in Sec.~\ref{sec:DIC} we can define:
\[
\bH = \bM \odot \left(\ones_2\ones_2^T \otimes \sum_t \sum_k \bSigma_{t,k}^T \odot \bSigma_{t,k} \right) 
\]
and
\[
\bE = \bM \odot \sum_t \sum_k (\ones_2\ones_2^T \otimes \bSigma_{t,k})^T \odot \bRh_{t,k}. 
\]
Similar to the unpolarized scenario, we express $\bJ^H\bJ$ using $\bH$ as
\[
\bJ^H\bJ = \begin{bmatrix}
	\bJ_{1}^H\bJ_{1} & \bJ_{1}^H\bJ_{2} \\
	\bJ_{2}^H\bJ_{1} & \bJ_{2}^H\bJ_{2}
\end{bmatrix}
\]
where
\[
\bJ_{q}^H\bJ_{p} = \begin{bmatrix}
\bD_{q,p} & \bG_p \bH \bG_q \\
\bG_p^* \bH \bG_q^* & \bD_{q,p}^*
\end{bmatrix},
\]
with $\bgt_{q,p} = \bg_q \odot \bg^*_p$, $\bD_{q,p} = \diag(\bH \bgt_{q,p})$ and $p,q \in \{1,2\}$.  The expression for the gradient, $\bgamma = [\bgamma_1^T, \bgamma_2^T]^T$, is the same as before with $\bgamma_q$ given by \eqref{eq:grad_ddc}
\[
\bgamma_q = \begin{bmatrix}
\bE_{q}\bg_{q}- \sum_{p=1}^2 \bG_p\bH \bgt_{q,p} \\
\bE_{q}^*\bg^*_{q}- \sum_{p=1}^2 \bG_p^*\bH \bgt^*_{q,p} 
\end{bmatrix}.
\]
When solving the LM update with $\lambda>0$ the only thing that changes is the definition of $\bD_{q,q}$ which is replaced by $\bD_{q,q} = \diag(\bH\bgt_{q,q} + \lambda \ones)$. Thus, the BGN method of section Sec.~\ref{sec:DDC} can easily be replaced by a block LM method without any increase of complexity. In fact, because adding $\lambda$ guarantees positive definiteness we don't need to calculate $\bzt$. However, if the convergence of this method is too slow, and/or calculating the exact inverse of a $P\times P$ matrix is affordable, it is better not to ignore the cross terms $\bJ_{1}^H\bJ_{2}$. In Appendix~\ref{app:BLDUDIPCUSK} we give a full description of finding the exact solution to \eqref{eq:LMupdates} using a Block-LDU decomposition. Having this exact method is also beneficial for the unpolarized direction-dependent method, as clusters could be updated in groups of 2 which should improve the convergence rate.

In conclusion, we have shown that a direction independent polarized calibration with unpolarized sky model can be achieved by solving an unpolarized problem with 2 directions having the same sky model. The only change needed is replacing the GN method by a LM method which can be recommended even for unpolarized problems. In similar way, we can transform a polarized direction-dependent algorithm  to an unpolarized calibration one with $2Q$ direction and shared sky-models. The CG method in Sec.~\ref{sec:DDC} is directly applicable without any modifications. In the next section, we show the convergence of these algorithms for direction independent case using a simulation.
 
\section{Simulations}
\subsection{Direction Independent Calibration}
In this section we use simulations to evaluate the performance of the proposed direction independent calibration technique. We simulate sample covariance data (visibilities) using the array configuration of the LOFAR radio telescope consisting only of the Dutch stations \cite{refIdLOFAR} with $P=62$. For the sky model we use the North Celestial Pole (NCP)\footnote{We would like to thank Sarod Yatawatta for this sky model.}. We use 5000 strongest component (point sources) in this field to generate both the data and construct the predicted sky model (i.e. $\bSigma_{t,k}$). We divide a typical LOFAR channel with 195.3 kHz bandwidth into $K=3$ sub channels of $\approx 65$kHz around the central frequency of 150 MHz. For each channel we generate $T=600$ snapshots, each with an integration time of 1 second, which translates into $N = 2\times 65 \times 10^3$ samples with Nyquist sampling. We repeat this for a total of 9 observations which are separated by 1 hour from each other. This is done in order to have enough rotation of the Earth to synthesize an image. Fig.~\ref{fig:ncpimg2} show an MVDR dirty image of the simulated field using 10 snapshot from each hour. Table~\ref{table:1} summarizes the computation on an Intel 7i-6700K CPU with 16GB of RAM. As we see, generating the predicted model $\bSigma_{t,k}$ is the most expensive part of the problem which cannot be avoided and is common among all currently available calibration models which use a sky model. By using $\bH$ and $\bE$ for the direction independent calibration we reduce the storage during the calibration by a factor of $TK = 1800$ and as is shown in fig.~\ref{fig:conv_di} and Table~\ref{table:1}, the algorithm converges very fast both in number of iterations and in computing time.

\begin{table*}
	\caption{Computation time DIC}
	\label{table:1}
	\centering
	\begin{tabular}{|l|c|c|c|}
		\hline
		&Generating $\bSigma_{t,k}$ & Calculating $\bH$ and $\bE$ & optimization \\
		\hline
		Unpolarized & 25s & 0.1s & 0.006s \\ 
		\hline
		Polarized CG &  &  & 0.064s\\ 
		\hline
	Polarized GN	 &  &  & 0.134s\\ 
		\hline
	Polarized FULL block-LDU LMA	 &  &  & 0.261s\\ 
		\hline
	\end{tabular}
\end{table*}

\begin{figure}
	\centering
	\includegraphics[width=0.45\textwidth]{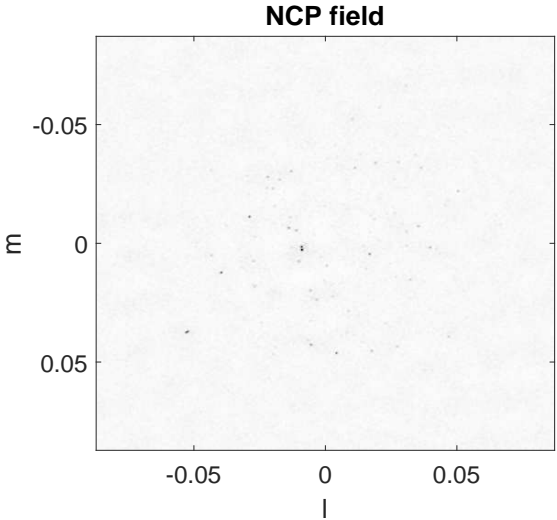}
	\caption{The MVDR dirty image of the simulated NCP field.}
	\label{fig:ncpimg2}
\end{figure}

\begin{figure}
	\centering
	\includegraphics[width=0.45\textwidth]{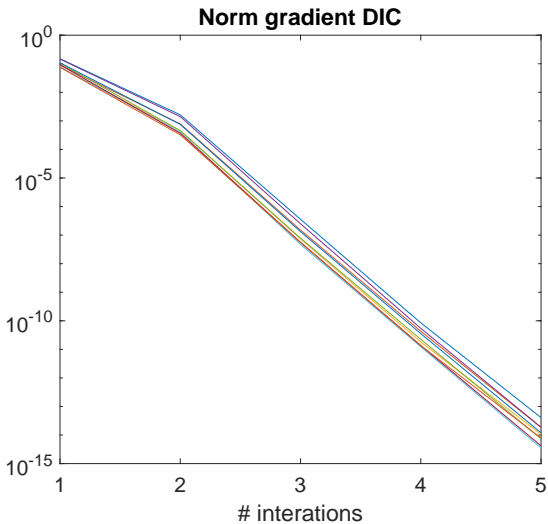}
	\caption{Convergence of the direction independent calibration for 9 calibration runs}
	\label{fig:conv_di}
\end{figure}

\begin{figure}
	\centering
	\includegraphics[width=0.45\textwidth]{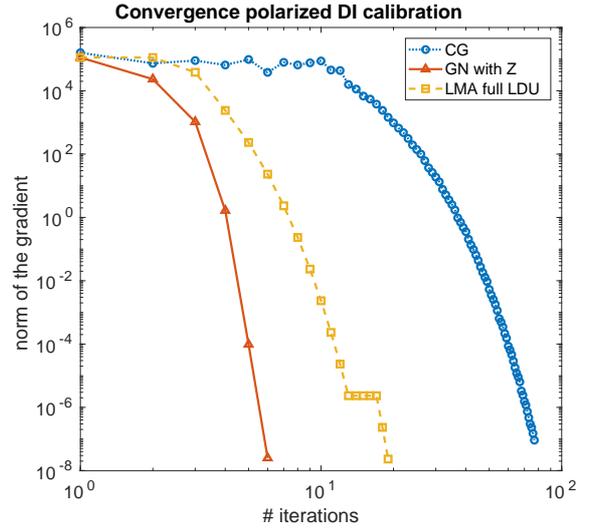}
	\caption{Convergence of the direction independent polarized calibration}
	\label{fig:conv_dipc}
\end{figure}

\begin{figure}
	\centering
	\includegraphics[width=0.45\textwidth]{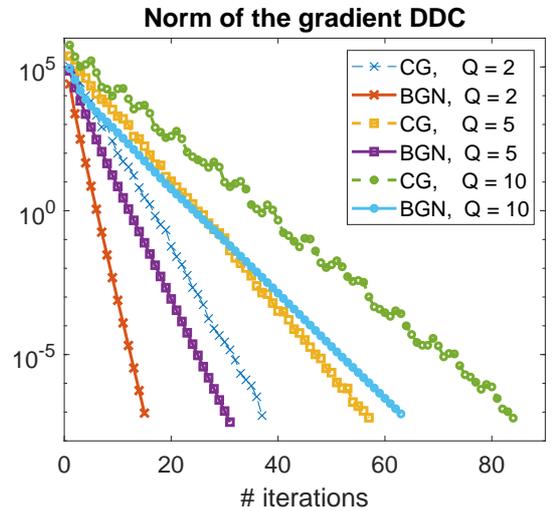}
	\caption{Convergence of the direction-dependent calibration}
	\label{fig:conv_ddc}
\end{figure}
\subsection{Direction Independent Polarized Calibration}
Because the assumed sky-model used in previous section consist only of unpolarized point sources, we can use it to simulate the polarization due to antenna gains similar to Sec.~\ref{sec:DIPOLCALUNM}. We will use a single 10 minute data set for this section.

Fig.~\ref{fig:conv_dipc} shows the convergence of three methods. The first method is direct use of the CG method from Sec.~\ref{sec:DDC} for $Q=2$, followed by the direct use of the Gauss-Newton method and LMA with full LDU decomposition as described in Appendix~\ref{app:BLDUDIPCUSK}. The gains are generated using Proper Gaussian complex random variables with unit mean and variance.

Table~\ref{table:1} shows the needed computation time for each of the methods. The steps in CG method are so cheap for $Q=2$ that the overall performance is very promising even though the convergence per iteration is slower than the other two methods. The GN method outperforms LMA because of the total number of iteration it needs in this simulation to converge. However, if we look at the computation time per iteration we see that they are similar. Because of the built in regularization of the LMA, in general we expect it to be more robust for polarized calibration.

\subsection{Direction Dependent Unpolarized Calibration}
For the direction dependent calibration we use again $P=62$ receivers with 20 randomly generated sources per direction. We use a single snapshot with $N=10k$ samples, which is moderate for radio astronomical observations. We then use the algorithm discussed in Sec.~\ref{sec:DDC} to find the gains. 

Fig.~\ref{fig:conv_ddc} shows the convergence speed of both algorithms based on the gradient. As expected the Block-Gauss-Newton (BGN) converges faster than the CG method. However, based on several repetition of our simulations, we have observed that the total computation time of the CG method, especially for larger $Q$, is much lower. For example in the case where $Q=10$ BGN method takes $\approx 7$s while CG take $\approx 0.3$s.  This fast convergence of CG is mainly because of the exact step-size calculations.
\section{Conclusion}
In this paper we have proposed new calibration algorithms for both direction dependent and independent calibration for radio interferometric array. We have shown that the optimal step-size can be calculated in a closed form fashion and does not require (expensive) line-search methods or approximations. All of the proposed algorithms converge reasonably fast and have very small storage requirements.

We also showed that if the sky model consist of un-polarized sources, the polarized direction independent gain calibration can be formulated as an un-polarized direction-dependent calibration with two directions, sharing the same sky-model. We then extended the calibration algorithms to accommodate this scenario.

There are several issues that are not addressed in this paper, including ionospheric effects and frequency dependency of the gains. The latter places additional restriction on the gains and hence will only improve the proposed algorithm without much anticipated modifications. These extensions are addressed in future works.
\bibliographystyle{IEEEtran}
\bibliography{biblio3}
\appendices
\section{Step Size}
\label{app:stepzie}
The least square costs function in Sec.\ \ref{sec:DDC} can be written as  
\[
\begin{array}{rl}
f(\btheta) &= \brh^H\brh - 2 \sum_q \bg_q^H \bE_q \bg_q + \sum_q \bgt_q^H \bH_{q,q}\bgt_q^H \\
& + 2 \sum_{q_1} \sum_{q2 > q_1} \bgt_{q_1,q_2}^H \bH_{{q_{1}},{q_{2}}} \bgt_{q_1,q_2}
\end{array}
\]
where $\bgt_q = \bg_{q} \odot \bg_q^*$ and $\bgt_{q_1,q_2} = \bg_{q_1} \odot \bg_{q_2}^*$. Let the direction of descent be $\bdelta$ and its sub-vector for $q$th direction be $\bdelta_q$, we are interested in
\[
\mu_{\text{opt}} = \arg \min_{\mu} f(\btheta + \mu \bdelta).
\]

Because of the quadratic relations in $\bgt$ we know that the cost function is a fourth order polynomial in $\mu$ with real coefficients, which means that finding optimal $\mu$ requires solving the roots of a third order polynomial. The gradient of the cost function with respect to $\mu$ is given by
\[
f'(\mu) = 4 c_1 \mu^3 + 3 c_2 \mu^2 + 2 c_3 \mu + c_1
\]
where $c_i  = a_i + b_i$ with
\[
\begin{array}{rl}
a_4 &=   \sum_q 2 \byt_q^H  \bH_{q,q} \bgt_q - 4  \Re\{\bdelta_q^H \bE_q \bg_q\},\\
a_3 &=  \sum_q 2 \bgt_q^H \bH_{q,q} \bxt_q + \byt_q^H \bH_{q,q} \byt_q - 2 \Re\{\bdelta_q^H  \bE_q  \bdelta_q\}, \\
a_2 &=  \sum_q 2 \byt_q^H \bH_{q,q}  \bxt_q,\\
a_1 &=  \sum_q \bxt_q^H \bH_{q,q} \bxt,\\

b_4 &=  \sum_q \sum_{p>q} 4 \Re\{\byt_{q,p}^H \bH_{q,p} \bgt_{q,p} \},\\
b_3 &=  \sum_q \sum_{p>q} 2 \Re\{\byt_{q,p}^H \bH_{q,p}  \byt_{q,p} + 2\bgt_{q,p}^H\bH_{q,p}  \bxt_{q,p} \},\\
b_2 &=  \sum_q \sum_{p>q} 4 \Re\{\byt_{q,p}^H \bH_{q,p} \bxt{q,p}\},\\
b_1 &= \sum_q \sum_{p>q}  2 \Re\{\bxt_{q,p}^H \bH_{q,p} \bxt_{q,p}\},\\

\byt_{q,p} &=  \bg_q \odot \bdelta_p^* + \bdelta_q \odot \bg_p^*,\\
\bxt_{q,p} &= \bdelta_q \odot \bdelta_p^*,
\end{array}
\]
$\byt_q=\byt_{q,q}$ and $\bxt_q=\bxt_{q,q}$.

\section{Block LDU for DI Polarized Calibration with Unpolarized Sky-Model }
\label{app:BLDUDIPCUSK}
In order to solve \eqref{eq:LMupdates} directly using a block-LDU decomposition we define 
\[
\bL_q^{-1} = \begin{bmatrix}
\bD_{q,q}^{-1/2} & \zeros \\
-\bX_q^{-1} \bC_q \bD_{q,q}^{-1/2}  & \bX_q^{-1}\bD_{q,q}^{-1/2} 
\end{bmatrix}
\]
where $\bD_{q,q} = \diag(\bH(\bg_q \odot \bg_q^*) + \lambda \ones)$, $\bX_q\bX_q^H = \bI - \bC_q\bC_q^H$ is a Cholesky decomposition and $\bC_q = \bD_{q,q}^{-1/2}\bG_q^*\bH\bG_q^*\bD_{q,q}^{-1/2}$. Let $\bb_q\equiv \bL_q^{-1} \bgamma_q$ and $\bC_{2,1} \equiv \bL_{2}^{-1}\bJ_2^H\bJ_{1}\bL_1^{-H}$ then the solution $\bdeltat_1$ is given by 
\begin{equation}
\label{eq:deltat1}
\bdeltat_1 = \begin{bmatrix}
\bL_1^{-H} & \zeros \\
\zeros &\bL_2^{-H}
\end{bmatrix}\begin{bmatrix}
\bx_1 \\
\bx_2
\end{bmatrix},
\end{equation}
with $\bx_1 = \bb_1-\bC_{2,1}^H\bx_2$ and $\bx_2$ the solution to
\[
(\bI-\bC_{2,1}\bC_{2,1}^H)\bx_2 = (\bb2 - \bC_{2,1}\bb_1).
\]
For $\lambda >0 $ the solutions $\bx_q$ must be unique and already in the correct augmented format.

Calculating the matrices $\bX_q^{-1}$ is the most computationally expensive part of this direct approach.
 
\end{document}